\documentclass[11pt]{article}
\usepackage[totalwidth=13.0cm,totalheight=20.0cm]{geometry}
\usepackage{latexsym,amsthm,amsmath,amssymb,url}
\usepackage[ruled, linesnumbered]{algorithm2e}

\usepackage{tikz,authblk}
\usetikzlibrary{decorations.pathreplacing}

\newtheorem{lemma}{Lemma}

\newtheorem{theorem}{Theorem}

\newtheorem{krule}{Reduction Rule}

\begin{document}

\title{Linear-Vertex Kernel for the Problem of Packing $r$-Stars into a Graph without Long Induced Paths}
\author[1]{Florian Barbero}
\author[2]{Gregory Gutin}
\author[2]{Mark Jones}
\author[2]{Bin Sheng}
\author[3,4]{Anders Yeo}
\affil[1]{Laboratoire d'Informatique, Robotique et Micro{\'e}lectronique de Montpellier, 161 rue Ada, 34095 Montpellier cedex 5, France}
\affil[2]{Department of Computer Science, Royal Holloway, University of London, TW20 0EX, Egham, Surrey, UK}
\affil[3]{Engineering Systems and Design, Singapore University of Technology and Design, Singapore} 
\affil[4]{Department of Mathematics,
      University of Johannesburg, 
      Auckland Park, 2006 South Africa}
\date{}
\maketitle

\begin{abstract}
\noindent  Let integers $r\ge 2$ and $d\ge 3$ be fixed. Let ${\cal G}_d$ be the set of graphs with no induced path on $d$ vertices. We study the problem of packing $k$ vertex-disjoint copies of $K_{1,r}$ ($k\ge 2$) into a graph $G$ from parameterized preprocessing, i.e., kernelization, point of view.  
We show that every graph $G\in {\cal G}_d$ can be reduced, in polynomial time, to a graph $G'\in {\cal G}_d$ with $O(k)$ vertices such that $G$ has at least $k$ vertex-disjoint copies of $K_{1,r}$ if and only if $G'$ has. Such a result is known for arbitrary graphs $G$ when $r=2$ and we conjecture that it holds for every $r\ge 2$.
\end{abstract}

\pagestyle{plain}

\section{Introduction}\label{sec:intro}
For a fixed graph $H$, the problem of deciding whether a graph $G$ has $k$ vertex-disjoint copies of $H$ is called  {\sc $H$-Packing}. The problem has many applications (see, e.g., \cite{BYHNSS02,BKGS01,KiHe78}), but unfortunately it is almost always intractable.  Indeed, Kirkpatrick and Hell \cite{KiHe78} proved that if $H$ contains a component with at least three vertices then {\sc $H$-Packing} is NP-complete. Thus, approximation, parameterised, and exponential algorithms have been studied for {\sc $H$-Packing} when $H$ is a fixed graph, see, e.g., \cite{BYHNSS02,Fel+11,Fel+05,PrSl06,WaNiFeCh08}.  

In this note, we will consider {\sc $H$-Packing} when $H=K_{1,r}$ and study  {\sc $K_{1,r}$-Packing}  from parameterized preprocessing, i.e., kernelization, point of view.\footnote{We provide basic definitions on parameterized algorithms and kernelization in the next section, for recent monographs, see \cite{CFLMPPS15,DF13}; \cite{Kra2014,LMS2012} are recent survey papers on kernelization.} Here $k$ is the parameter. As a parameterized problem, {\sc $K_{1,r}$-Packing}  was first considered by Prieto and Sloper \cite{PrSl06} who obtained an $O(k^2)$-vertex kernel for each $r\ge 2$ and a kernel with at most $15k$ vertices for $r=2$. (Since the case $r=1$ is polynomial-time solvable, we may restrict ourselves to $r\ge 2$.)  The same result for $r=2$ was proved by Fellows {\em et al.} \cite{Fel+11} and it was improved to $7k$ by Wang {\em et al.} \cite{WaNiFeCh08}. 

Fellows {\em et al.} \cite{Fel+11} note that, using their approach, the bound of \cite{PrSl06} on the number of vertices in a kernel for any $r\ge 3$ can likely be improved to subquadratic. We believe that, in fact, there is a linear-vertex kernel for every $r\ge 3$ and we prove Theorem \ref{th:ker} to support our conjecture. A path $P$ in a graph $G$, is called {\em induced} if it is an induced subgraph of $G$. For an integer $d\ge 3$, let ${\cal G}_d$ denote the set of all graphs with no induced path on $d$ vertices. 

\begin{theorem}\label{th:ker}
Let integers $r\ge 2$ and $d\ge 3$ be fixed. Then {\sc $K_{1,r}$-Packing} restricted to graphs in ${\cal G}_d$, has a kernel with $O(k)$ vertices. 
\end{theorem}
 Since $d$ can be an arbitrary integer larger than two, Theorem \ref{th:ker} is on an ever increasing class of graphs which, in the ``limit", coincides with all graphs. To show that Theorem \ref{th:ker} is an optimal\footnote{If $K_{1,r}$-{\sc Parking} was polynomial time solvable, then it would have a kernel with $O(1)$ vertices.} result, in a sense, we prove that {\sc $K_{1,r}$-Packing} restricted to graphs  in ${\cal G}_d$ is ${\cal NP}$-hard already for $d=5$ and every fixed $r\ge 3$:  
 
 \begin{theorem}\label{th:NP}
 Let $r \geq 3$. It is ${\cal NP}$-hard to decide if the vertex set of a graph in ${\cal G}_5$ can be partitioned into vertex-disjoint copies of $K_{1,r}$. 
\end{theorem}

We cannot replace ${\cal G}_5$ by ${\cal G}_4$ (unless ${\cal NP}={\cal P}$) due to the following assertion, whose proof is given in the Appendix.

\begin{theorem}\label{th:poly}
 Let $r \geq 3$ and $G\in {\cal G}_4$. We can find the maximal number of vertex-disjoint copies of $K_{1,r}$ in $G$ in polynomial time. 
\end{theorem}

\section{Terminology and Notation}

For a graph $G$, $V(G)$ ($E(G)$, respectively) denotes the vertex set (edge set, respectively) of $G$, $\Delta(G)$ denotes the maximum degree of $G$ and $n$ its number of vertices. For a vertex $u$ and a vertex set $X$ in $G$, $N(u)=\{v: uv\in E(G)\}$, $N[u] = N(u)\cup \{u\},$ $d(u)=|N(u)|$, $N_X(u) = N(u) \cap X$, $d_X(u)=|N_X(u)|$ and $G[X]$ is the subgraph of $G$ induced by $X$. We call $K_{1,r}$ an {\em $r$-star}. 
We say a star {\em intersects} a vertex set if the star uses a vertex in the set.
 We use $(G,k,r)$ to denote an instance of the $r$-star packing problem. If there are $k$ vertex-disjoint $r$-stars in $G$, we say $(G,k,r)$ is a {\sc Yes}-instance, and we write $G \in \star(k,r)$. 
 Given disjoint vertex sets $S,T$ and integers $s,r$, we say that $S$ has $s$ $r$-{\em stars in} $T$ if there are $s$ vertex-disjoint $r$-stars with centers in $S$ and leaves in $T$.

A \emph{parameterized problem} is a subset $L\subseteq \Sigma^* \times
\mathbb{N}$ over a finite alphabet $\Sigma$. A parameterized problem $L$ is
\emph{fixed-parameter tractable} if the membership of an instance
$(I,k)$ in $\Sigma^* \times \mathbb{N}$ can be decided in time
$f(k)|I|^{O(1)}$ where $f$ is a computable function of the
{\em parameter} $k$ only.
Given a parameterized problem $L$,
a \emph{kernelization of $L$} is a polynomial-time
algorithm that maps an instance $(x,k)$ to an instance $(x',k')$ (the
\emph{kernel}) such that $(x,k)\in L$ if and only if
$(x',k')\in L$ and $k'+|x'|\leq g(k)$ for some
function $g$. 
It is well-known that a decidable parameterized problem $L$ is fixed-parameter
tractable if and only if it has a kernel. Kernels of small size are of
main interest, due to applications.

\section{Proof of Theorem \ref{th:ker}}

Note that the $1$-star packing problem is the classic maximum matching problem and if $k = 1$, the $r$-star packing problem is equivalent to deciding whether $\Delta(G) \ge r$. 
Both of these problems can be solved in polynomial time. Henceforth, we assume $r,k>1$.

A vertex $u$ is called a \textit{small vertex} if $\max \{d(v):v\in N[u]\} < r$. A graph without a small vertex is a \textit{simplified graph}.

We now give two reduction rules for an instance $(G,k,r)$ of {\sc $K_{1,r}$-Packing}.

\begin{krule}\label{rule:small-vertices}
 If graph $G$ contains a small vertex $v$, then return the instance $(G - v, k,r)$.
\end{krule}

It is easy to observe that Reduction Rule~\ref{rule:small-vertices} can be applied in polynomial time. 

\begin{krule}\label{rule:obstacle-removal}
 Let $G = (V,E)$ be a graph and let $C, L$ be two vertex-disjoint subsets of $V$. The pair $(C, L)$ is called a \textit{constellation} if $G[C \cup L] \in \star(|C|,r)$ and there is no star $K_{1,r}$ intersecting $L$ in the graph $G[V \setminus C]$. If $(C, L)$ is a constellation, return the instance $(G[V\setminus (C\cup L)], k-|C|)$.
\end{krule}

It is easy to observe that Reduction Rule~\ref{rule:obstacle-removal} can be applied in polynomial time, provided we are given a suitable constellation.

\begin{lemma} Reduction~Rules~\ref{rule:small-vertices} and~\ref{rule:obstacle-removal} are safe.\end{lemma}

\begin{proof}
Clearly, a small vertex $v$ can not appear in any $r$-star.
Therefore Reduction Rule \ref{rule:small-vertices} is safe as $G$ and $G-v$ will contain the same number of $r$-stars.

To see that Reduction Rule~\ref{rule:obstacle-removal} is safe, it is sufficient to show
that $G \in \star(k,r)$ if and only if $G[V\setminus(C\cup L)] \in \star(k-|C|,r)$. 
On the one hand, if $G[V\setminus(C\cup L)] \in \star(k-|C|,r)$, the hypothesis $G[C \cup L] \in \star(|C|,r)$ implies $G \in \star(k,r)$. On the other hand, there are at most $|C|$ vertex-disjoint stars intersecting $C$. But by hypothesis, every star intersecting $L$ also intersects $C$. We deduce that there are at most $|C|$ stars intersecting $C \cup L$, and so if $G \in \star(k,r)$, there are at least $k-|C|$ stars in $G[V-(C\cup L)]$: $G[V\setminus(C\cup L)] \in \star(k-|C|,r)$.
\end{proof}

Note that as both rules modify a graph by deleting vertices, any graph $G'$ that is derived from a graph $G \in {\cal G}_d$ by an application of Rules~\ref{rule:small-vertices} or~\ref{rule:obstacle-removal} is also in ${\cal G}_d$.


Recall the Expansion Lemma, which is a generalization of the well-known Hall's theorem.

\begin{lemma}(\textbf{Expansion Lemma})\cite{FoLoMiPhSa11}
Let $r$ be a positive integer, and let $m$ be the size of the maximum matching in a bipartite graph $G$ with vertex bipartition
$X\cup Y$. If $|Y|>rm$, and there are no isolated vertices in $Y,$ then there exist nonempty vertex sets
$S\subseteq X, T \subseteq Y$ such that $S$ has $|S|$ $r$-stars in $T$ and no vertex in $T$ has a neighbor outside $S$. 
Furthermore, the sets $S, T$ can be found in polynomial time in the size of $G$.
\end{lemma}

Henceforth, we will use the following modified version of the expansion lemma.

\begin{lemma}(\textbf{Modified Expansion Lemma})
Let $r$ be a positive integer, and let $m$ be the size of the maximum matching in a bipartite graph $G$ with vertex bipartition $X\cup Y$. If $|Y|>rm$, and there are no isolated vertices in $Y,$ then there exists a polynomial algorithm(in the size of $G$) which returns a partition $X=A_1\cup B_1$, $Y=A_2\cup B_2$, such that 
$B_1$ has $|B_1|$ $r$-stars in $B_2$,
 $E(A_1,B_2)=\emptyset$, and $|A_2|\leq r|A_1|$.
\end{lemma}
\begin{proof}
If $|Y|\leq rm$, then we may return $A_1 = X$, $A_2=Y$, $B_1=B_2=\emptyset$, as $m \leq |X|$ and hence $|Y|\leq r|X|$.
Otherwise, apply the Expansion Lemma to get nonempty vertex sets
$S\subseteq X, T \subseteq Y$ such that $S$ has $|S|$ $r$-stars in $T$ and no vertex in $T$ has a neighbor in $Y$ outside $S$.
Let $X'=X\setminus S$ and $Y'=Y \setminus T$. If  $G[X' \cup Y']$ has isolated vertices in $Y'$, move all of them from $Y'$ to $T$. 
If $|Y'|\leq r|X'|$, we may return $A_1 = X'$, $A_2=Y'$, $B_1=S,$ and $B_2=T$.

So now assume $|Y'|> r|X'|$. 
In this case, apply the algorithm recursively on $G[X' \cup Y']$  to get a partition $X'=A_1'\cup B_1', Y' = A_2'\cup B_2'$, such that $B_1'$ has $|B_1'|$ stars in $B_2'$, $E(A_1',B_2')=\emptyset$, and $|A_2'|\leq r|A_1'|$.
Then return $A_1 = A_1'$, $B_1 = B_1' \cup S$, $A_2 = A_2'$, $B_2 = B_2' \cup T$. Observe that $B_1$ has $|B_1'|+|S| = |B_1|$ stars in $B_2$, $E(A_1, B_2) \subseteq E(A_1',B_2') \cup E(X\setminus S,T) = \emptyset$, and $|A_2| = |A_2'| \leq r|A_1'| = r|A_1|$, as required.
As each iteration reduces $|X|$ by at least $1$, we will have to apply less than $|X|+|Y|$ iterations, each of which uses at most one application of the Expansion Lemma, and so the algorithm runs in polynomial time.
\end{proof}




\noindent{\bf Proof of Theorem \ref{th:ker}.}
By exhaustively applying Reduction Rule~\ref{rule:small-vertices}, we may assume we have a simplified graph.
Let $G$ be a simplified graph in ${\cal G}_d$.
Now find a maximal $r$-star packing of the graph $G$ with $q$ stars.
We may assume $q < k$ as otherwise we have a trivial {\sc Yes}-instance.
Let $S$ be the set of vertices in this packing, and let $D = V(G)\setminus S$.

For any $u \in D$, let $D[u]$ be
the set of vertices $v \in D$ for which there is a path from $v$ to $u$ using only vertices in $D$ - 
that is, $D[u]$ is the the set of vertices in the component of $G[D]$ containing $u$.
As our star-packing is maximal, $d_D(v) < r$ for every $v \in D$.
As $G \in {\cal G}_d$, 
every $v \in D[u]$ has a path to $u$ in $G[D]$ with at most $d-1$ vertices
(as otherwise the shortest path in $G[D]$ from $v$ to $u$ is an induced path on at least $d$ vertices).
It follows that $|D[u]| \le 1 + r + r^2 + \dots + r^{d-1} \le r^{d}$.


We will now find a partition of $S$ into $Big(S) \cup Small(S)$, and $D$ into $B(D) \cup U(D)$, such that $|B(D)| \le r^{d+1}|Small(S)|$, and either $Big(S)=U(D) = \emptyset$ or $(Big(S), U(D))$ is a constellation. 
As $|Small(S)|\le |S| \le (r+1)k$, it follows that either $|V(G)|\le (r+1)k + (r+1)r^{d+1}k$, or we can apply Reduction Rule \ref{rule:obstacle-removal} on $(Big(S), U(D))$.



%


We will construct $Big(S), Small(S),B(D),U(D)$ algorithmically as described below. 
Throughout, we will preserve the properties that %
\begin{enumerate}
 \item $|B(D)|\le|Small(S)|r^{d+1}$, 
 \item $U(D)$ has no neighbors in $Small(S) \cup B(D)$.
\end{enumerate}

Initially, set $Big(S)=S, U(D)=D, Small(S)=B(D)=\emptyset$.


 
While $|U(D) \cap N(Big(S))| > r|Big(S)|$, do the following.
 
 If there is a vertex $u \in Big(S)$ such that $|N(u)\cap U(D)| < r$, let $X = \bigcup\{D[v]: v \in N(u) \cap U(D)\}$. Observe that as $|D[v]| \le r^{d}$ for all $v \in D$, $|X|< r^{d+1}$.
 Now set $Small(S) = Small(S) \cup \{u\}$, $Big(S)=Big(S)\setminus \{u\}$, $B(D) = B(D) \cup X$, $U(D) = U(D)\setminus X$. 
 It follows that Property 1 is preserved.
 Note that no vertex in the new $U(D)$ has a neighbor in $X$  (as all neighbors of $X$ in $D$ lie in $X$).
 Similarly no vertex in the new $U(D)$ is adjacent to $u$ (as such a vertex would be in the old $U(D)$ and so would have been added to $X$).
 Therefore there are still no edges between the new $U(D)$ and the new $Small(S) \cup B(D)$,
 and so Property 2 is preserved.


Otherwise (if every vertex $u \in Big(S)$ has $|N(u)\cap U(D)| \ge r$), let $H$ denote the maximal bipartite subgraph of $G$ with vertex partition $Big(S)\cup (U(D) \cap N(Big(S))$,
and apply the Modified Expansion Lemma to $H$. We will get a partition $Big(S) = A_1 \cup B_1$ and $U(D) \cap N(Big(S)) = A_2 \cup B_2$ such that
 $E(A_1,B_2) = \emptyset, |A_2| \le r|A_1|$ 
 and $B_1$ has $|B_1|$ $r$-stars in $B_2$.

If the Modified Expansion Lemma returns $B_1=Big(S)$, then we claim that $(Big(S), U(D))$ is a constellation.
To see this, firstly note that $|Big(S)|$ has $|Big(S)|$ $r$-stars in $U(D)$.
Secondly, note that since we chose the vertices of a maximal star packing for $S$, there is no $r$-star contained in $G[U(D)]$.
As $U(D)$ has no neighbors in $Small(S) \cup B(D)$, it follows that there is no $r$-star intersecting $U(D)$ in $G \setminus Big(S)$.
Thus $(Big(S), U(D))$ is a constellation, and the claim is proved.
In this case the algorithm stops.




So now assume that the Modified Expansion Lemma returns $Big(S)=A_1 \cup B_1$ with $A_1 \neq \emptyset$.
Let $X = \bigcup \{D[v]: v \in N(A_1) \cap U(D)\}$.
Note that as $E(A_1,B_2) = \emptyset$ and $|A_2|\le r|A_1|$,
we have $|X| \le |\bigcup \{D[v]: v \in A_2\}| \le |A_2|r^{d} \le |A_1|r^{d+1}$.
Then let $Small(S) = Small(S) \cup A_1$, $Big(S) = Big(S) \setminus A_1$, $B(D) = B(D) \cup X$, $U(D) = U(D) \setminus X$.
Note that after this move, we still have that $|B(D)|\le |Small(S)|r^{d+1}$, 
and $U(D)$ has no neighbors in $Small(S) \cup B(D)$.



Note that in either case, $|Big(S)|$ strictly decreases, so the algorithm must eventually terminate, 
either because $(Big(S), U(D))$ is a constellation, or because $|U(D) \cap N(Big(S))| \leq r|Big(S)|$.
If $(Big(S), U(D))$ is a constellation,
apply Reduction Rule~\ref{rule:obstacle-removal} using $(Big(S), U(D))$. This gives us a partition in which $Big(S)=U(D)=\emptyset$.
Thus in either case, we have that $|U(D) \cap N(Big(S))| \leq r|Big(S)|$.
Note that every vertex $u \in U(D)$ is in $D[v]$ for some $v \in N(S)$ (as otherwise, either $\max \{d(v):v\in N[u]\} < r$ or $G[D]$ contains an $r$-star, a contradiction in either case). Moreover such a $v$ must be in $U(D) \cap N(Big(S))$, as there are no edges between $U(D)$ and $Small(S) \cup B(D)$.
Thus $|U(D)|\le r^d|U(D) \cap N(Big(S))| \leq r^{d+1}|Big(S)|$.
Then we have $|V(G)|=|S|+|U(D)|+|B(D)| \leq |S| + r^{d+1}|Big(S)| + r^{d+1}|Small(S)| \le (r^{d+1}+1)|S| \le (k-1)(r+1)(r^{d+1}+1) = O(k)$.



\section{Proof of Theorem \ref{th:NP}}

A {\em split graph} is a graph where the vertex set can be partitioned into a clique and an independent set.

An instance of the well-known ${\cal  NP}$-hard problem {\sc $3$-Dimensional Matching} contains a vertex set that can be partitioned into three equally large sets $V_1,V_2,V_3$ (also called partite sets). 
Let $k$ denote the size of each of $V_1,V_2,V_3$.
It furthermore contains a number of $3$-sets containing exactly one vertex from each $V_i$, $i=1,2,3$.
The problem is to decide if there exists a set of $k$ vertex disjoint $3$-sets (which would then cover all vertices).
Such a set of $k$ vertex disjoint $3$-sets is called a perfect matching.
The $3$-sets are also called {\em edges} (or {\em hyperedges}).

\begin{theorem} \label{thm1}
Let $r \geq 3$. It is ${\cal NP}$-hard to decide if the vertex set of a split graph can be partitioned into vertex disjoint copies of $K_{1,r}$.
\end{theorem}
\begin{proof}
We will reduce from {\sc $3$-Dimensional Matching}. Let ${\cal I}$ be an instance of $3$-dimensional matching. 
Let $V_1,V_2,V_3$ denote the three partite sets of ${\cal I}$ and let $E$ denote the set of
edges in ${\cal I}$. Let $m = |E|$ and $k=|V_1|=|V_2|=|V_3|$.
We will build a split graph $G_{\cal I}$ as follows. 
Let $V = V_1 \cup V_2 \cup V_3$ be the vertices of ${\cal I}$.
Let $X_1$ be a set of $m$ vertices and $X_2$ be a set of $m-k$ vertices and let $X = X_1 \cup X_2$.
Let $Y$ be a set of $(m-k)(r-1)$ vertices and let $W$ be a set of $k(r-3)$ vertices (if $r=3$ then $W$ is empty).
Let the vertex set of $G_{\cal I}$ be $V \cup X \cup Y \cup W$.

Add edges such that $X$ becomes a clique in $G_{\cal I}$. Let each vertex in $X_1$ correspond to a distinct edge in $E$ and connect that vertex with the $3$ vertices in $V$ which 
belongs to the corresponding edge in $E$. Furthermore add all edges from $X_1$ to $W$. Finally, for each vertex in $X_2$ add $r-1$ edges to $Y$ in such a way that 
each vertex in $Y$ ends up with degree one in $G_{\cal I}$. This completes the construction of $G_{\cal I}$.

Clearly $G_{\cal I}$ is a split graph as $X$ is a clique and $V \cup Y \cup W$ is an independent set. We will now show that the vertex set of $G_{\cal I}$ can be partitioned into 
vertex disjoint copies of $K_{1,r}$ if and only if ${\cal I}$ has a perfect matching.

First assume that ${\cal I}$ has a perfect matching. Let $E' \subseteq E$ denote the edges of the perfect matching. For the vertices in $X_1$ that correspond to the edges in $E'$
we include the three edges from each such vertex to $V$ as well as $r-3$ edges to $W$.  This can be done such that we obtain $k$ vertex disjoint copies of $K_{1,r}$ covering all of $V$ and $W$ as well as
$k$ vertices from $X_1$. Now for each vertex in $X_2$ include the $r-1$ edges to $Y$ as well as one edge to an unused vertex in $X_1$. This can be done such that 
we obtain an additional $m-k$ vertex disjoint copies of $K_{1,r}$. We have now constructed $m$ vertex disjoint copies of $K_{1,r}$ which covers all the vertices in $G_{\cal I}$, as required.

Now assume that the vertex set of $G_{\cal I}$ can be partitioned into vertex disjoint copies of $K_{1,r}$. 
As $|V \cup W \cup Y \cup X| = m(r+1)$ we note that we have $m$ vertex disjoint copies of $K_{1,r}$, which we will denote by ${\cal K}$.
As all vertices in $Y$ need to be included in such copies we note that 
every vertex of $X_2$ is the center vertex of a $K_{1,r}$. Let ${\cal K}'$ denote these $m-k$ copies of $K_{1,r}$.
Each $K_{1,r}$ in ${\cal K}'$ must include $1$ edge from $X_2$ to $X_1$. These $m-k$ edges form a matching, implying
that $m-k$ vertices of $X_1$ also belong to the copies of $K_{1,r}$ in ${\cal K}'$. 
This leaves $k$ vertices in $X_1$ that are uncovered and $rk$ vertices in $V \cup W$ that are uncovered.
Furthermore, as $V \cup W$ is an independent set, each copy of $K_{1,r}$ in ${\cal K} \setminus {\cal K}'$ must contain a vertex of $X_1$.
As $|{\cal K} \setminus {\cal K}'| = k$ we note that 
the $k$ copies of $K_{1,r}$  in  ${\cal K} \setminus {\cal K}'$ must include exactly one vertex from $X_1$. 
Also as each vertex in $X_1$ has exactly three neighbours in $V$, each such
$K_{1,r}$ also contains $3$ vertices from $V$ (as $V$ needs to be covered) and therefore $r-3$ vertices form $W$.  
Therefore the $k$ vertices in $X_1$ that belong to copies of $K_{1,r}$ in ${\cal K} \setminus {\cal K}'$ correspond
to $k$ edges in $E$  which form a perfect matching in $G_{\cal I}$.

This completes the proof as we have shown that $G_{\cal I}$ can be partitioned into
vertex disjoint copies of $K_{1,r}$ if and only if ${\cal I}$ has a perfect matching.
\end{proof}

The following lemma is known. We give the simple proof for completeness.

\begin{lemma} \label{lem1}
No split graph contains an induced path on 5 vertices.
\end{lemma}
\begin{proof}
Assume $G$ is a split graph where $V(G)$ is partitioned into an independent set $I$ and a clique $C$.
For the sake of contradiction assume that $P= p_0 p_1 p_2 p_3 p_4$ is an induced $P_5$ in $G$. 
As $I$ is independent we note that $\{p_0,p_1\} \cap C \not= \emptyset$ and 
$\{p_3,p_4\} \cap C \not= \emptyset$. As $C$ is a clique there is therefore an edge from 
a vertex in $\{p_0,p_1\}$ to a vertex in $\{p_3,p_4\}$. This edge implies that $P$ is not an induced $P_5$ in $G$, a contradiction.
\end{proof}

\noindent{\bf Proof of Theorem \ref{th:NP}.}
By Lemma~\ref{lem1}, ${\cal G}_5$ contains all split graphs.
The result now follows immediately from Theorem~\ref{thm1}.~\qed

\vspace{3mm}

\noindent{\bf Acknowledgment.} Research of GG was partially supported by Royal Society Wolfson Research Merit Award. Research of BS was supported by China Scholarship Council.

\newpage

\appendix

\section{Proof of Theorem \ref{th:poly} }

Note that ${\cal G}_4$ is the family of {\em cographs} \cite{BLS99}. It is well-known \cite{BLS99} that any non-trivial (i.e., with at least two vertices) cograph
$G$ is either disconnected or its complement is disconnected. Below let $n$ denote the order of $G$ and let $m$ denote the size of $G$.
The following lemma is well-known.

\begin{lemma}\label{lem:Comp}
For any graph $G$, we can in time 
$O(n^2)$ find the connected components of $G$ and the connected components of the complement of $G$.
\end{lemma}

\begin{lemma}\label{lem:Max} 
For any $G \in {\cal G}_4$ and any $s \geq 1$ we can in time
$O(n^2)$   find a set of $s$ vertices, say $S$, in $G$ such that $|N[S]|$ is maximum possible.
\end{lemma}

\begin{proof}
 Let $C_1,C_2,\ldots,C_l$ be the connected components of $G$ ($l \geq 1$). Assume first that all the components are non-trivial. 
 As any induced subgraph of a cograph is also a cograph we note that 
 the complement of each $C_i$ is disconnected.
 Therefore for each $i=1,2,\ldots,l$ there exists a non-trivial (each part is non-empty) 
 partition $(X_i,Y_i)$ of $V(C_i)$ such that all edges exist between $X_i$ and $Y_i$ in $G$.
 Let $m_i$ be maximum degree of a vertex in $C_i$ for each $i=1,2,\ldots,l.$
 
 The maximum number of vertices we can add to $N[S]$ by adding one vertex from $C_i$ is $m_i+1$ and the maximum number of vertices added to 
$N[S]$ by adding two vertices from $C_i$ is $|V(C_i)|$ as we can add a vertex from $X_i$ and one from $Y_i$.
Therefore the maximum possible $|N[S]|$ is the sum of the $s$ largest numbers in the set $m_1+1,m_2+1,\ldots,m_s+1,(|V(C_1)|-m_1-1), 
(|V(C_2)|-m_2-1),\ldots , (|V(C_l)|-m_l-1)$.
Furthermore it is easy to find the actual set $S$.\\
It is not hard to modify the proof above for the case when some $C_i$'s are trivial.
\end{proof}

Now we are ready to prove the main result of this appendix.\\

\noindent{\bf Proof of Theorem \ref{th:poly}:}
Let $G \in {\cal G}_4$ and let $r \geq 3$ be arbitrary.
First assume that $G$ is connected, which implies that the complement of $G$ is disconnected.
Let $X$ and $Y$ partition $V(G)$ such that all edges exist between $X$ and $Y$ in $G$.
We now consider two cases.\\

\noindent{\bf Case 1: } $|X| > r|Y|$ or $|Y| > r|X|$. Without loss of generality, assume that $|X| > r|Y|$.
In this case we recursively find the maximum number of $r$-stars we can pack into $G[X]$.
Let $m_x$ be the maximum number of $r$-stars in $G[X]$. 
If $(r+1)m_x + (r+1)|Y| \leq n$, then the optimal answer is that we can pack $m_x + |Y|$
$r$-stars into $G$ as we can always find $|Y|$ $r$-stars with centers in $Y$ and not touching the
$m_x$ $r$-stars we already found in $G[X]$. 
If  $(r+1)m_x + (r+1)|Y| > n$, then the optimal solution is $\lfloor n/(r+1) \rfloor$ $r$-stars as
we can pick $|Y|$ $r$-stars touching as few of the $m_x$ $r$-stars in $G[X]$ as possible and then pick as many of the 
$m_x$ $r$-stars that are left untouched. This completes this case.\\

\noindent{\bf Case 2: } $|X| \leq r|Y|$ and $|Y| \leq r|X|$. Let $x = |X|$ and $y=|Y|$ and define $a$ and $b$ as follows:

\[
a = \frac{ry-x}{r^2-1} \mbox{    and    } b = \frac{rx-y}{r^2-1}
\]

Let $a'=\lfloor a \rfloor = a - \epsilon_a $ and $b' = \lfloor b \rfloor = b - \epsilon_b$. 
We will first show that we can find $a'+b'$ $r$-stars such that 
$a'$ of the $r$-stars have the center in $X$ and all leaves in $Y$ and $b'$ of the 
$r$-stars have the center in $Y$ and all leaves in $X$. 
This is possible due to the following:

\[
\begin{array}{rcl}
a'r + b' & = & (a - \epsilon_a) r + (b-\epsilon_b) \\ 
 & = & r \frac{ry-x}{r^2-1} + \frac{rx-y}{r^2-1} - (r \epsilon_a + \epsilon_b) \\
 & = & y - (r \epsilon_a + \epsilon_b) \\
\end{array}
\]

And, analogously,

\[
b'r + a' =  x - (r \epsilon_b + \epsilon_a) \\
\]

As $0 \leq \epsilon_a < 1$ and $0 \leq \epsilon_b < 1$ we note that we cover all vertices in $G$ except $r \epsilon_a + \epsilon_b + r \epsilon_b + \epsilon_a 
= (r+1)(\epsilon_a+\epsilon_b)$. Therefore the number of vertices we cannot cover by the $r$-stars above is strictly less than $2(r+1)$.
If $(r+1)(\epsilon_a+\epsilon_b) < r+1$ then we have an optimal solution (covering all vertices except at most $r$), 
so assume that $(r+1)(\epsilon_a+\epsilon_b) \geq r+1$.

Clearly the optimal solution is either $a'+b'$ or $a'+b'+1$.
As we already have a solution with $a'+b'$ $r$-stars we will now determine if there is a solution with $a'+b'+1$ $r$-stars.

If some vertex, say $w_x$, in $X$ has degree at least $r$ in $G[X]$, then there is indeed a solution with $a'+b'+1$ $r$-stars, because of the following.
As  $(r+1)(\epsilon_a+\epsilon_b) \geq r+1$ we must have $\epsilon_a>0$ and $\epsilon_b>0$, which implies that we can pick an $r$-star 
with center in $w_x \in X$ and with at most $r\epsilon_b + \epsilon_a-1$ leaves in $X$ and at most $r \epsilon_a + \epsilon_b$ leaves in $Y$. Once this
$r$-star has been picked it is not difficult to pick an additional $a'$ $r$-stars with centers in $X$ (and leaves in $Y$) and $b'$
 $r$-stars with centers in $Y$ (and leaves in $X$), due to the above.  Therefore we may assume no vertex in $X$ has degree at least $r$ in $G[X]$.
Analogously we may assume that no vertex in $Y$ has degree at least $r$ in $G[Y]$.

If there exists $a'+1$ vertices $S_X$ in $X$ such that $|N[S_X] \cap X| \geq a'+1+ r- (r \epsilon_a + \epsilon_b) $, then proceed as follows.
We can create $a'+1$ stars in $G[X]$ such that they together have exactly $r - (r \epsilon_a + \epsilon_b)$ non-centers. By the above each star has 
less than $r$ leaves, so we can expand these $a'+1$ stars to $r$-stars by adding leaves from $Y$. This uses up  $a'+1+ r- (r \epsilon_a + \epsilon_b)$
vertices from $X$ and $(a'+1)r - (r - (r \epsilon_a + \epsilon_b))$ vertices from $Y$. Adding an additional $b'$ stars with the center in $Y$ and 
all leaves in $X$ uses up $b'$ vertices from $Y$ and $rb'$ vertices from $X$.
Therefore we have used $b'+a'r + (r \epsilon_a + \epsilon_b) = y$ vertices from $Y$ and the following number of vertices from $X$,

\[
a'+rb' + 1 + r - (r \epsilon_a + \epsilon_b) = x - (r \epsilon_b + \epsilon_a) + 1 +r -  (r \epsilon_a + \epsilon_b) = x + 1 + r - (r+1)(\epsilon_a+\epsilon_b)
\]

As $(r+1)(\epsilon_a+\epsilon_b) \geq r+1$ we note that we use at most $x$ vertices from $X$ and we have a solution with $a'+b'+1$ $r$-stars.
Analogously if there exists $b'+1$ vertices $S_Y$ in $Y$ such that $|N[S_Y] \cap Y| \geq b'+1+ r- (r \epsilon_b + \epsilon_a) $, we obtain 
$a'+b'+1$ $r$-stars. By applying Lemma~\ref{lem:Max} to $G[X]$ and $G[Y]$ we can decide the above in polynomial time.

We may therefore assume that no such $S_X$ or $S_Y$ exist. We will now show that $a'+b'$ is the optimal solution. For the sake of contradiction 
assume that we have $a^*$ $r$-stars with centers in $X$ and $b^*$ $r$-stars with centers in $Y$, such that they are vertex disjoint and
$a^*+b^* = a'+b'+1$. Without loss of generality we may assume that $a^* \geq a'+1$. 
The $a^*$ $r$-stars with centers in $X$ all have at least one leaf in $Y$ as the maximum degree in $G[X]$ is less than $r$. 
Furthermore by the above ($S_X$ does not exist) any $a'+1$ $r$-stars with centers in $X$ have more than
$r(a'+1) - ( r - (r \epsilon_a + \epsilon_b))$ leaves in $Y$. Therefore we use strictly more than the following number of vertices in $Y$.

\[
r(a'+1) - ( r - (r \epsilon_a + \epsilon_b)) + (a^*+b^* - (a'+1)) =  ra' + r \epsilon_a + \epsilon_b + b' = y
\]

This contradiction implies that the optimal solution is $a'+b'$ in this case. This completes the case when $G$ is connected.

Finally assume that $G$ is disconnected. In this case we recursively solve the problem for each connected component, which can be added together to get an optimal solution for $G$.
It is not difficult to see that the above can be done in polynomial time.\qed


\begin{thebibliography}{1}

\bibitem{BYHNSS02} R. Bar-Yehuda, M. Halld{\'o}rsson, J. Naor, H. Shachnai, and I. Shapira, Scheduling split intervals, in 30th Annu. ACM-SIAM Symp. on
Discrete Algorithms, 2002, pp. 732--741.

\bibitem{BKGS01} R. Bejar, B. Krishnamachari, C. Gomes, and B. Selman, Distributed constraint satisfaction in a wireless sensor tracking system, Workshop on
Distributed Constraint Reasoning, Internat. Joint Conf. on Artificial Intelligence, 2001.

\bibitem{BLS99} A. Brandst{\"a}dt, V.B. Le, and  J.P. Spinrad. Graph Classes: A Survey, SIAM, 1999.

\bibitem{CFLMPPS15} M. Cygan, F.V. Fomin, L. Kowalik, D. Lokshtanov, D. Marx, M. Pilipczuk, M. Pilipczuk, and S. Saurabh,  Parameterized Algorithms, Springer, 2015.

\bibitem{DF13} R.G. Downey and M.R. Fellows, Foundations of Parameterized Complexity, Springer, 2013.

\bibitem{Fel+11} M. Fellows, J. Guo, H. Moser, and R. Niedermeier, J. Comput. Syst. Sci. 77:1141--1158, 2011.

\bibitem{Fel+05} M. Fellows, P. Heggernes, F. Rosamond, C. Sloper, and J.A. Telle, Finding $k$ Disjoint Triangles in an Arbitrary Graph. In WG'05, Lect. Notes Comput. Sci. 3353:235--244, 2005


\bibitem{FoLoMiPhSa11} F.V. Fomin, D. Lokshtanov, N. Misra, G. Philip, and S. Saurabh, Hitting forbidden minors: Approximation and Kernelization. In STACS 2011, LIPIcs 9:189--200, 2011.

\bibitem{GN2007} J. Guo and R. Niedermeier. Linear problem kernels for NP-hard problems on
planar graphs. In ICALP 2007, Lect. Notes Comput. Sci. 4596:375-386, 2007.

\bibitem{KiHe78} D.G. Kirkpatrick and P. Hell, On the completeness of a generalized matching problem. In 10th STOC, ACM Symposium on Theory of Computing, 240--245, 1978.  

\bibitem{Kra2014} S. Kratsch, Recent developments in kernelization: A survey. Bulletin EATCS, no. 113, 2014.

\bibitem{LMS2012} D. Lokshtanov, N. Misra, and S. Saurabh, Kernelization - preprocessing with a guarantee.  Lect. Notes Comput. Sci. 7370:129-161, 2012.

\bibitem{PrSl06} E. Prieto and C. Sloper, Looking at the stars, Theor. Comput. Sci. 351:437--445, 2006.

\bibitem{WaNiFeCh08} J. Wang, D. Ning, Q. Feng, and J. Chen, An improved parameterized algorithm for a generalized matching problem, In
TAMC'08, Lect. Notes Comput. Sci., 4978:212--222, 2008.
 
\end{thebibliography}
\end{document}